\documentclass[a4paper,12pt]{article}
\usepackage{amssymb}
\usepackage{amsfonts}
\usepackage{mathrsfs}
\usepackage{amsthm}
\usepackage{amsmath}
\usepackage[a4paper,hmargin={1.4in},top=1.2in,bottom=1.2in ]{geometry}
\usepackage[colorlinks=false,pdfborder=000]{hyperref}

\makeatletter
\renewcommand\@biblabel[1]{#1\hspace{0.5em}}
\makeatother

\newcommand{\rmnum}[1]{\romannumeral #1}
\newcommand{\Rmnum}[1]{\expandafter\@slowromancap\romannumeral #1@}

 \newtheorem{thm}{Theorem}[section]

\newtheorem{lem}[thm]{Lemma}
 \newtheorem{prop}[thm]{Proposition}
\newtheorem{cor}[thm]{Corollary}

  \newcommand{\f}{\mathbb{F}_{q}}

\begin{document}
\title{New Deep Holes of Generalized Reed-Solomon Codes}
\author{Jun Zhang, Fang-Wei Fu\\
\emph{\small Chern Institute of Mathematics and LPMC,} \\
\emph{\small  Nankai University, Tianjin, 300071, China.}\\
\emph{\small Email: zhangjun04@mail.nankai.edu.cn,} \\
\emph{\small      fwfu@nankai.edu.cn}\\
\mbox{}\\
 Qun-Ying Liao\\
 \emph{\small Institute of Mathematics and Software Science,}\\
\emph{\small  Sichuan
Normal University,} \\
\emph{\small Chengdu, 610066, China.}\\
\emph{\small Email: liao\underline{ }qunying@yahoo.com.cn}}
\date{}
\maketitle
\begin{abstract}
Deep holes play an important role in the decoding of generalized Reed-Solomon codes.
Recently, Wu and Hong~\cite{WH} found a new class
of deep holes for standard Reed-Solomon codes. In the present paper, we give a concise method to obtain a new class of deep holes for generalized Reed-Solomon codes. In particular, for standard Reed-Solomon codes, we get the new class of deep holes given in~\cite{WH}.

Li and Wan~\cite{L.W1} studied deep holes of generalized
Reed-Solomon codes $GRS_{k}(\f,D)$ and characterized deep holes
defined by polynomials of degree $k+1$. They showed that this
problem is reduced to be a subset sum problem in finite fields.
Using the method of Li and Wan, we obtain some new deep holes for
special Reed-Solomon codes over finite fields with even
characteristic. Furthermore, we study deep holes of the extended
Reed-Solomon code, i.e., $D=\f$ and show polynomials of degree $k+2$
can not define deep holes.

 %While we note that by a simple computation in~\cite{L.W1}, we obtain new deep holes for special standard Reed-Solomon codes over finite field of even %characteristic. And similarly, we give some new deep holes for some generalized Reed-Solomon codes.
\begin{flushleft}
\emph{keywords:} coding theory, Reed-Solomon codes, list decoding, deep holes, multiplicative character, quadratic equation.
\end{flushleft}
\end{abstract}

\section{Introduction}

Let $\f$ be the finite field of $q$ elements with
characteristic $p$. Fix a subset
$D=\{x_1,\ldots,x_n\}\subseteq \f$, which is
called the evaluation set. The generalized Reed-Solomon code
$C=GRS_{k}(\f,D)$ of length $n$ and dimension $k$ over
$\f$ is defined to be
$$
GRS_{k}(\f,D)=\{(f(x_1),\ldots,f(x_n))\in
\mathbb{F}^n_q| f(x)\in \f[x], \deg f(x)\leq k-1\}.
$$
Its elements are called codewords. The most widely used cases are
$D=\f$ or $\f^{*}$. For the case $D=\f^{*}$,
it is just the standard Reed-Solomon code. For the case $D=\f$,
it is the extended Reed-Solomon code.

For any word $u\in \f^{n}$, by the Lagrange interpolation, there is a
polynomial $f$ of degree $\leqslant n-1$ such that
\[
   u=u_{f}=(f(x_{1}),f(x_{2}),\cdots,f(x_{n})).
\]
Clearly, $u_{f}\in GRS_{k}(\f,D)$ if and only if $\deg(f)\leqslant
k-1$. We also say that $u_{f}$ is defined by the polynomial $f(x)$.

Let $C$ be an $[n,k,d]$ linear code over $\f$. The \emph{error distance} of any word $u\in\f^n$ to $C$ is defined to be
\[d(u,C)=\min\{d(u,v)\,|\,v\in C\}.
\]
where
\[
  d(u,v)=\#\{i\,|\,u_{i}\neq v_{i},\,1\leqslant i\leqslant n\}
\]
is the Hamming distance between words $u$ and $v$. The error distance play an important role in the list decoding of Reed-Solomon codes.
 Given a received word $u\in \f^{n}$, if the error distance
is small, say, $d(u,C)\leq n-\sqrt{nk}$, then the list decoding
algorithm of Sudan \cite{S} and Guruswami-Sudan \cite{G.S} provides
a polynomial time algorithm for the decoding of $u$. When the error
distance increases, the maximal likelihood decoding becomes more
complicated, in fact, \textbf{NP-}complete for generalized
Reed-Solomon codes \cite{G.V}.

The most important algorithmic problem in coding theory is the
maximal likelihood decoding: given a word $u\in \mathbb{F}^n_q$,
find a codeword $v\in C$ such that $d(u,v)=d(u,C)$. The decision
version of this problem is essentially computing the error distance
$d(u,C)$ for a received word $u$. For the generalized Reed-Solomon
codes $GRS_{k}(\f,D)$, the following result is well-known.

\begin{lem}[\cite{L.W1}]\label{liwan}
For any $k \leqslant \deg(f)\leqslant n-1$, we have
\[
   n-\deg(f)\leqslant d(u_{f},C)\leqslant n-k.
\]
\end{lem}

In particular, the word $u$ is called a \emph{deep hole} if the
above upper bound is attained, i.e., $d(u,\mathcal{C})=n-k$.
Following Lemma~\ref{liwan}, we see that the vectors
defined by polynomials of degree $k$ are deep holes.

For the standard Reed-Solomon code $GRS_{k}(\f,\f^{*})$, based on
numerical computations, Cheng and Murray~\cite{C.M} conjectured that
vectors defined by polynomials of degree $k$ are the only deep holes
possible. As a theoretical evidence, they proved that their
conjecture is true for words $u_{f}$ defined by polynomial $f$ if
$d=\deg(u_{f})-k$ is small and $q$ is sufficiently large compared to
$d+k$. For those words defined by polynomials in $\f[x]$ of low
degrees, Li and Wan \cite{L.W2} applied the method of Cheng and Wan
\cite{C.W1} to study the error distance $d(u,C)$ for the standard
Reed-Solomon code. Liao \cite{L} extended the results in~\cite{L.W2}
to those words defined by polynomials in $\f[x]$ of high degrees.
Recently, by means of a deeper study of the geometry of
hypersurfaces, Cafure and et al.~\cite{Cafure} made some improvement
of the results in~\cite{L.W2}.

Recently, Wu and Hong~\cite{WH} found a new class
of deep holes for standard Reed-Solomon codes. They considered the standard
Reed-Solomon code $GRS_{k}(\f,\f^{*})$ as a cyclic code with the generator polynomial
\[
g(x)=(x-\alpha)(x-\alpha^{2})\cdots(x-\alpha^{n-k}),
\]
where $\alpha$ is a primitive element of $\f$. They found two classes of deep holes, namely
$a\frac{g(x)}{x-\alpha}+l(x)g(x)$ and
$a\frac{g(x)}{x-\alpha^{n-k}}+l(x)g(x)$ for any $a\in \f^{*}$ and
$l(x)\in \f[x]$ with $\deg(l(x))\leqslant k-1$. After the inverse
discrete Fourier transform, the latter is just the trivial ones
defined by polynomials of degree $k$, but the former is new. The
new deep holes are of the form $ux^{q-2}+l(x)$ for any $u\neq 0$ and
$l(x)\in \f[x]$ with $\deg(l(x))\leqslant k-1$.
And then they conjectured that there are no more deep holes for standard Reed-Solomon codes.

In Section~2, we study deep holes of generalized Reed-Solomon codes. Using a simple method, we give a new class of deep holes for any generalized Reed-Solomon code in the case $D\neq \f$. In Section~3, using the explicit formula given in~\cite{L.W2}, we obtain new deep holes for special Reed-Solomon codes over finite fields with even characteristic, which are not contained in the two classes of deep holes given in Section~2.
In Section~4, we study deep holes of extended Reed-Solomon codes, and characterize the deep holes defined by polynomials of degree $k+2$. We show that the existence of these deep holes is equivalent to non-existence of solutions for a class of quadratic equations. We prove the quadratic equation always has solutions. Then there are no deep holes defined by polynomials of degree $k+2$.

\section{A New class of Deep Holes for Generalized Reed-Solomon Codes}
In this section, we consider generalized Reed-Solomon codes $GRS_{k}(\f,D)$ over the finite field $\f$. We give a new class of deep holes for $GRS_{k}(\f,D)$ besides the trivial ones defined by the polynomials of degree $k$.

In the case that $D=\f^*$, i.e., standard Reed-Solomon codes, Wu and Hong considered them as cyclic codes and gave a new class of deep holes defined by polynomials $f=ax^{q-2}+f_{\leqslant k-1}(x)\,(a\neq 0)$ where $f_{\leqslant k-1}(x)$ represents the terms of degree $\leqslant k-1$ in the polynomial $f$. But for the general valuation set $D\neq \f^*$, the generalized Reed-Solomon code $GRS_{k}(\f,D)$ can not be considered as a cyclic code, so their method is invalid. Below we give a simple method to show that for any $b\notin D$, these words defined by polynomials $f=a(x-b)^{q-2}+f_{\leqslant k-1}(x)\,(a\neq 0)$ are still deep holes for generalized Reed-Solomon codes $GRS_{k}(\f,D)$.

First, we give the key lemma.
\begin{lem}\label{keylem}
Let $\f$ be a finite field with $q$ elements and $\f[x]$ the ring of polynomials over $\f$. For any $D\subsetneqq \f$, any $b\notin D$, and any polynomial $(x-b)^{q-2}-g(x)\in \f[x]$  with $\deg(g(x))\leqslant k-1$, $a\in D$ is a zero of $(x-b)^{q-2}-g(x)$ if and only if $a$ is a zero of $1-(x-b)g(x)$. So the polynomial $(x-b)^{q-2}-g(x)$ has at most $k$ zeros in $D$.
\end{lem}
\begin{proof}
Notice that the polynomial $(x-b)^{q-2}-g(x)$ has the same zeros as the polynomial $(x-b)^{2}((x-b)^{q-2}-g(x))=(x-b)-(x-b)^{2}g(x)$ on $D$ for $b\notin D$. While the latter has the same zeros as the polynomial $1-(x-b)g(x)$ on $D$. So the polynomial $(x-b)^{q-2}-g(x)$ has the same zeros as the polynomial $1-(x-b)g(x)$ on $D$. But the degree of $1-(x-b)g(x)$ is not larger than $k$, so all the mentioned polynomials has at most $k$ zeros in $D$.
\end{proof}
By this lemma, we easily obtain a new class of deep holes for $GRS_{k}(\f,D)$ besides the trivial ones defined by polynomials of degree $k$.
\begin{thm}\label{newdh}
Let $\f$ be a finite field with $q$ elements and $C=GRS_{k}(\f,D)$ the generalized Reed-Solomon code over $\f$, where $D\subsetneqq \f$ is the evaluation set with cardinality $n$. Then for any $b\notin D$, polynomials $f=a(x-b)^{q-2}+f_{\leqslant k-1}(x)\,(a\neq 0)$ define deep holes for $GRS_{k}(\f,D)$.
\end{thm}
\begin{proof}
Without loss of generality, we assume $D\subseteq\f^*$, $b=0$, and we may only consider the polynomial $f=x^{q-2}$.
By Lemma~\ref{keylem}, for any polynomial $g(x)\in \f[x]$ with $\deg(g(x))\leqslant k-1$, the polynomial $f-g$ has at most $k$ zeros in $D$. So the error distance
\[
   d(u_{f},C)=n-\max_{g\in \f[x],\,\deg(g)\leqslant k-1}\left|\{\textrm{zeros of $f-g$}\}\right|\geqslant n-k.
\]
On the other hand, for any $S\subseteq D$ with $|S|=k$, set
\[
   g_{S}=\frac{1-a_{S}\prod_{\gamma\in S}(x-\gamma)}{x},\,\textrm{where $a_{S}=(-1)^{k}\prod_{\gamma\in S}\gamma^{-1}$.}
\]
Then $g_{S}\in \f[x]$ has degree $k-1$ and the polynomial
\[
   1-xg_{S}(x)=a_{S}\prod_{\gamma\in S}(x-\gamma)
\]
has $k$ zeros in $D$. From the argument in the proof of Lemma~\ref{keylem}, the polynomial $f-g_{S}=x^{q-2}-g_{S}$ has $k$ zeros in $D$. So the lower bound for $d(u_{f},C)$ can be attained, i.e.,
\[
 d(u_{f},C)=n-k.
\]
In other words, the polynomial $f=x^{q-2}$ defines a deep hole for the generalized Reed-Solomon code $C=GRS_{k}(\f,D)$.
\end{proof}
In particular, if we take $D=\f^*$, then by Theorem~\ref{newdh} the polynomials $f=ax^{q-2}+f_{\leqslant k-1}(x)\,(a\neq 0)$ define deep holes for the standard Reed-Solomon code $GRS_{k}(\f,\f^*)$, which were given by Wu and Hong~\cite{WH}.
\section{A New class of Deep Holes for Special Reed-Solomon Codes over Finite Fields with Even Characteristic}
In this section, we consider Reed-Solomon codes $GRS_{k}(\f,D)$ over the finite field $\f$ with even characteristic. And we give a new class of deep holes for some Reed-Solomon codes in this case. Li and Wan~\cite{L.W2} characterized the deep holes defined by polynomials of degree $k+1$ as follows.

Let $f=x^{k+1}-bx^{k}+f_{\leqslant k-1}(x)\in\f[x]$. By Lemma~\ref{liwan}, $n-k-1\leqslant d(u_{f},C)\leqslant n-k$.
If $f$ does not define a deep hole, i.e., $d(u_{f},C)=n-k-1$, it is equivalent to that there exists $g\in \f[x],
\,\deg(g)\leqslant k-1$ such that
\[
  x^{k+1}-bx^{k}+f_{\leqslant k-1}(x)-g(x)=\prod_{j=1}^{k+1}(x-a_{i_{j}})
\]
for a $(k+1)$-subset
$\{a_{i_{1}},\cdots,a_{i_{k+1}}\}\subseteq D$. It is also
equivalent to
\[
  a_{i_{1}}+\cdots+a_{i_{k+1}}=b
\]
for a $(k+1)$-subset
$\{a_{i_{1}},\cdots,a_{i_{k+1}}\}\subseteq D$.

Following the notation in~\cite{L.W1}, for any subset $D$ of
$\f^{*}$, let
\[
   N(t,b,D)=\#\{S\subseteq D\,|\,\sum_{x\in S}x=b,\,\#S=t\}.
\]
Then the polynomial $f=x^{k+1}-bx^{k}+f_{\leqslant k-1}(x)$ defines a deep hole if and only if
 $$N(t,b,D)=0.
 $$
Taking $D=\f^{*}$ or $D=\f^{*}\setminus \{1\}$, Li and Wan presented the explicit formula for $N(t,b,D)$.
\begin{prop}[\cite{L.W1}]
Let $\f$ be a finite field with $q$ elements.

\emph{(\textbf{\rmnum{1})}} For $D=\f^{*}$,
\[
   N(t,0,\f^{*})=\frac{1}{q}\binom{q-1}{t}+(-1)^{t+\lfloor
   t/p\rfloor}\frac{q-1}{q}\binom{q/p-1}{\lfloor t/p\rfloor}.
\]
\vskip2mm \emph{(\textbf{\rmnum{2})}} For $D=\f^{*}\setminus \{1\}$, let
\[
 R_{t}=-p(-1)^{\lfloor t/p\rfloor}\binom{q/p-2}{\lfloor
 t/p\rfloor}+(p-1-<t>_{p})(-1)^{\lfloor t/p\rfloor}\binom{q/p-1}{\lfloor
 t/p\rfloor}
\]
and
\[
 M(t,b)=-(-1)^{\lfloor t/p\rfloor}\binom{q/p-2}{\lfloor
 t/p\rfloor}+\delta_{b,t}(-1)^{\lfloor t/p\rfloor}\binom{q/p-1}{\lfloor
 t/p\rfloor},
\]
where $<t>_{p}$ denotes the least non-negative residue of $t$ modulo
$p$, and $\delta_{b,t}=1$ if $<b>_{p}$ is greater than $<t>_{p}$ and
$\delta_{b,t}=o$ otherwise. Then
\[
  N(t,b,D)=\frac{1}{q}\binom{q-2}{t}+\frac{1}{q}R_{t}-(-1)^{t}M(t,t-b).
\]
\end{prop}
Thanks to the explicit formulae, the following is immediate for the case that the characteristic $p=2$ and $q>4$.
For $D=\f^{*}$ or $D=\f^{*}\setminus \{1\}$,
\[
     N(q-3,0,D)=0.
\]
Then $d(u_{f},C)=n-k$. Therefore, we obtain a class of new deep holes for the Reed-Solomon code $GRS_{q-4}(\f,D)$.
\begin{thm}\label{evenchar}
Let $\f$ be a finite field with characteristic $2$ and $D=\f^{*}$ or $D=\f^{*}\setminus \{1\}$. If $q>4$, then the
vectors defined by polynomials
$ax^{q-3}+f_{\leqslant k-1}(x)\,(a\neq 0)$ are deep holes for the Reed-Solomon code $GRS_{q-4}(\f,D)$.
\end{thm}
Wu and Hong conjectured that there are only two classes of deep holes for the standard Reed-Solomon code $GRS_{k}(\f,\f^{*})$ defined by polynomials of the form
$ax^{k}+f_{\leqslant k-1}(x)$ and $ax^{q-2}+f_{\leqslant k-1}(x)$ ($a\neq 0$). While in case for the special standard Reed-Solomon code $GRS_{q-4}(\f,\f^{*})$ over the finite field $\f$ with even characteristic, Theorem~\ref{evenchar} gives a counterexample of the conjecture of Wu and Hong~\cite{WH}, i.e., deep holes defined by polynomials $ax^{q-3}+f_{\leqslant k-1}(x)\,(a\neq 0)$.

\section{Deep Holes defined by Polynomials of Degree $k+2$}

In this section, we consider the extended
Reed-Solomon code $GRS_{k}(\f,\f)$ over the finite field $\f$ of $q>5$ elements with odd characteristic. For the polynomial of the form $f=x^{k+2}-ax^{k+1}-bx^{k}+f_{\leqslant k-1}(x)$, by Lemma~\ref{liwan}, we have $n-k-2\leqslant d(u_{f},C)\leqslant n-k$. Then $d(u_{f},C)\leqslant n-k-1$, i.e., such $f$ does not define a deep hole, if and only if there exists some polynomial $g\in \f[x]$ of degree $\deg(g)\leqslant k-1$ such that
\[
   f-g=(x-\gamma)\prod_{\beta\in S}(x-\beta),
\]
for some $S\subseteq \f$ with cardinality $k+1$, and some $\gamma \in \f$. This implies that
\begin{displaymath}
\left\{\begin{array}{rcl}
a & =  & \gamma+\sum_{\beta\in S}\beta, \\
b & = & -\sum_{\{\beta_{1},\beta_{2}\}\subset S}\beta_{1}\beta_{2}-\gamma\sum_{\beta\in S}\beta.\\
\end{array}\right.
\end{displaymath}
Hence we have the following lemma.
\begin{lem}
The polynomial $f=x^{k+2}-ax^{k+1}-bx^{k}+f_{\leqslant k-1}(x)\in \f[x]$ can not define a deep hole if and only if there are some $S\subseteq \f$ with cardinality $k+1$, and some $\gamma \in \f$ such that
\begin{displaymath}
\left\{\begin{array}{rcl}
a & =  & \gamma+\sum_{\beta\in S}\beta, \\
b & = & -\sum_{\{\beta_{1},\beta_{2}\}\subset S}\beta_{1}\beta_{2}-\gamma\sum_{\beta\in S}\beta.\\
\end{array}\right.
\end{displaymath}
\end{lem}
From the two equalities, we have
$$b=-\sum_{\{\beta_{1},\beta_{2}\}\subset S}\beta_{1}\beta_{2}-(a-\sum_{\beta\in S}\beta)\sum_{\beta\in S}\beta.
$$
Hence, to search all such $f$ that does not define a deep hole, we need to find all $a\in \f$ and $b\in \f$ such that the equation
\begin{displaymath}
\left\{\begin{array}{rcl}
b & = & -\sum_{1\leqslant i < j\leqslant k+1}X_{i}X_{j}+(\sum_{i=1}^{k+1}X_{i}-a)\sum_{i=1}^{k+1}X_{i},\\
X_{i}&\neq & X_{j},\qquad \textrm{for all $i\neq j$,}
\end{array}\right.
\end{displaymath}
i.e.,
\begin{displaymath}
\left\{\begin{array}{rcl}
b & = & \sum_{1\leqslant i < j\leqslant k+1}X_{i}X_{j}+\sum_{i=1}^{k+1}X_{i}^{2}-a\sum_{i=1}^{k+1}X_{i},\\
X_{i}&\neq & X_{j},\qquad \textrm{for all $i\neq j$}
\end{array}\right.
\end{displaymath}
has solutions in $\f$.

Indeed, we will show that the equation always has solutions for any $a,b\in\f$. Let $t=k+1$, so we have
\begin{thm}\label{main}
For any $a,b\in\f$, for all $3\leqslant t\leqslant q-2$, the equation
\begin{displaymath}
\left\{\begin{array}{rcl}
b & = & \sum_{1\leqslant i < j\leqslant t}X_{i}X_{j}+\sum_{i=1}^{t}X_{i}^{2}-a\sum_{i=1}^{t}X_{i},\\
X_{i}&\neq & X_{j},\qquad \textrm{for all $i\neq j$}
\end{array}\right.
\end{displaymath}
has solutions in $\f$.
\end{thm}
The proof we give is highly nontrivial, so we present the proof in
the Appendix.

Since $\deg(f)=k+2\leqslant n-1=q-1$, we only consider $k\leqslant q-3$. By Theorem~\ref{main}, we obtain the main result of this section.
\begin{thm}
Let $\f$ be a finite field of $q>5$ elements and $GRS_{k}(\f,\f)$ the standard
Reed-Solomon code. For all $2\leqslant k\leqslant q-3$, the polynomials $f=ux^{k+2}+ax^{k+1}+bx^{k}+f_{\leqslant k-1}(x)\in \f[x]$ ($u\in \f^*,a,b\in\f$) do not define deep holes.
\end{thm}

\section{Conclusion}
The result of Section~4 supports the conjecture of Cheng and
Murray~\cite{C.M} for extended Reed-Solomon codes. And it is easy to
see that as the evaluation set $D$ becomes small, there will be more
deep holes for the Reed-Solomon code $GRS_{k}(\f,D)$. In particular,
in Section~2 we have seen that there is a new class of deep holes
whenever there is an element $a\in\f\setminus D$. For finite fields
$\f$ with cardinality $q>5$ and $D\subseteq \f$ with cardinality $q$
or $q-1$, Li and Wan~\cite{L.W1} proved that there is no deep hole
of $GRS_{k}(\f,D)$ defined by polynomials of degree $k+1$ when
$2<k<q-3$. In Section~4, we have seen that if the finite field $\f$
has odd characteristic, there is also no deep hole of
$GRS_{k}(\f,\f)$ defined by polynomials of degree $k+2$ when
$2\leqslant k\leqslant q-3$. Similarly, in this case, we can prove
that there is also no deep hole of $GRS_{k}(\f,\f^*)$ defined by
polynomials of degree $k+2$ when $2\leqslant k<q-4$.

\section*{Acknowledgements}
%The authors thank Prof. Wan D.Q. for a lot of comments.

The first two authors are supported by the National
Science Foundation of China (Nos. 61171082, 10990011, 60872025). The third author is supported by the National
Science Foundation of China (No. 10990011), the Ph.D. Programs
Foundation of Ministry of Education of China (No. 20095134120001) and Sichuan Provincial Advance Research
Program for Excellent Youth Leaders of Disciplines in Science of
China (No. 2011JQ0037).

%------------reference---------------------------------

\appendix
\section{Appendix}
To complete this paper, we give the proof of Theorem~\ref{main} in this section. We only prove the theorem for the case that $q=p>2$ a prime integer. In the proof, we also need to assume $q\geqslant257$. For prime $q<257$, we can use the computer to check it.

 For general prime power $q$, the proof is similar. Note that the reduction in Lemma~\ref{red} is not valid when the characteristic $p|(t+1)$. But for small $t$ (i.e., $t<c_{1}q$ for some constant $c_{1}$), the method in the proof of Lemma~\ref{thm} (and in the part (b) of the proof for Theorem~\ref{main}) still works; for large $t$ (i.e., $t>c_{2}q$ for some constant $c_{2}$), using Lemma~\ref{opp}, it is enough to consider the complement set $\f\setminus\{X_{1},\cdots,X_{t}\}$. Together with the required version of Theorem~\ref{dual} (taking $c=1/2$ and a proper $\epsilon$ in Theorem~5.3 in~\cite{liwan2}), we can finish the proof of Theorem~\ref{main}.

First, we give some lemmas. Due to the next lemma, it is reduced to the case $a=0$.
\begin{lem}\label{red}
For any $a,b\in\f$, the equation
\begin{displaymath}
\left\{\begin{array}{rcl}
b & = & \sum_{1\leqslant i < j\leqslant t}X_{i}X_{j}+\sum_{i=1}^{t}X_{i}^{2}-a\sum_{i=1}^{t}X_{i},\\
X_{i}&\neq & X_{j},\qquad \textrm{for all $i\neq j$}
\end{array}\right.
\end{displaymath}
has solutions in $\f$, if and only if for any $\beta\in\f$, the equation
\begin{displaymath}
\left\{\begin{array}{rcl}
\beta & = & \sum_{1\leqslant i < j\leqslant t}X_{i}X_{j}+\sum_{i=1}^{t}X_{i}^{2},\\
X_{i}&\neq & X_{j},\qquad \textrm{for all $i\neq j$}
\end{array}\right.
\end{displaymath}
has solutions in $\f$.
\end{lem}
\begin{proof}The necessity part holds obviously.

For the other direction, for any $a,b\in\f$, taking $\beta=\frac{ta^{2}}{2(t+1)}+b$, suppose $(x_{1},\cdots,x_{t})\in\f^t$ is a solution of the equation
\begin{displaymath}
\left\{\begin{array}{rcl}
\beta & = & \sum_{1\leqslant i < j\leqslant t}X_{i}X_{j}+\sum_{i=1}^{t}X_{i}^{2},\\
X_{i}&\neq & X_{j},\qquad \textrm{for all $i\neq j$.}
\end{array}\right.
\end{displaymath}
Set $c=\frac{a}{t+1}$, then
\begin{displaymath}
\begin{array}{cl}
 & \sum_{1\leqslant i < j\leqslant t}(x_{i}+c)(x_{j}+c)+\sum_{i=1}^{t}(x_{i}+c)^{2}-a\sum_{i=1}^{t}(x_{i}+c)\\
=& \sum_{1\leqslant i < j\leqslant t}x_{i}x_{j}+\sum_{i=1}^{t}x_{i}^{2}+((t+1)c-a)\sum_{i=1}^{t}x_{i}-tac+\frac{t(t+1)}{2}c^{2}\\
=&b.
\end{array}
\end{displaymath}
 So $(x_{1}+c,\cdots,x_{t}+c)\in\f^t$ is a solution of the equation
 \begin{displaymath}
\left\{\begin{array}{rcl}
b & = & \sum_{1\leqslant i < j\leqslant t}X_{i}X_{j}+\sum_{i=1}^{t}X_{i}^{2}-a\sum_{i=1}^{t}X_{i},\\
X_{i}&\neq & X_{j},\qquad \textrm{for all $i\neq j$.}
\end{array}\right.
\end{displaymath}
\end{proof}

\begin{lem}\label{opp}
For any $S\subset\f$ with cardinality $2\leqslant |S|\leqslant q-2$, set $S'=\f\setminus S$. Then
\[
  \sum_{\{\beta_{1},\beta_{2}\}\subset S}\beta_{1}\beta_{2}+\sum_{\beta\in S}\beta^2=\sum_{\{\gamma_{1},\gamma_{2}\}\subset S'}\gamma_{1}\gamma_{2}.
\]
\end{lem}
\begin{proof}
 We have
 \begin{displaymath}
\begin{array}{rcl}
0& =  & \sum_{\{\beta_{1},\beta_{2}\}\subset \f}\beta_{1}\beta_{2}+\sum_{\beta\in \f}\beta^2 \\
 & = & (\sum_{\{\beta_{1},\beta_{2}\}\subset S}\beta_{1}\beta_{2}+\sum_{\beta\in S}\beta^2)+(\sum_{\{\gamma_{1},\gamma_{2}\}\subset S'}\gamma_{1}\gamma_{2}+\sum_{\gamma\in S'}\gamma^2)\\
 &  & +\sum_{\beta\in S,\gamma\in S'}\beta\gamma\\
 & = & (\sum_{\{\beta_{1},\beta_{2}\}\subset S}\beta_{1}\beta_{2}+\sum_{\beta\in S}\beta^2)+(\sum_{\{\gamma_{1},\gamma_{2}\}\subset S'}\gamma_{1}\gamma_{2}+\sum_{\gamma\in S'}\gamma^2)\\
  &&+(\sum_{\gamma\in S'}\gamma)(-\sum_{\gamma\in S'}\gamma)\\
  & = & \sum_{\{\beta_{1},\beta_{2}\}\subset S}\beta_{1}\beta_{2}+\sum_{\beta\in S}\beta^2-\sum_{\{\gamma_{1},\gamma_{2}\}\subset S'}\gamma_{1}\gamma_{2}.
\end{array}
\end{displaymath}
So
\[
  \sum_{\{\beta_{1},\beta_{2}\}\subset S}\beta_{1}\beta_{2}+\sum_{\beta\in S}\beta^2=\sum_{\{\gamma_{1},\gamma_{2}\}\subset S'}\gamma_{1}\gamma_{2}.
\]
\end{proof}
By this lemma, it is enough to consider the equation
\begin{displaymath}
\left\{\begin{array}{rcl}
b & = & \sum_{1\leqslant i < j\leqslant t}X_{i}X_{j},\\
X_{i}&\neq & X_{j},\qquad \textrm{for all $i\neq j$.}
\end{array}\right.
\end{displaymath}
\begin{lem}\label{thm}
For any $b\in\f$ and $2\leqslant t\leqslant \frac{q-1}{2}$, the equation
\begin{displaymath}
\left\{\begin{array}{rcl}
b & = & \sum_{1\leqslant i < j\leqslant t}X_{i}X_{j},\\
X_{i}&\neq & X_{j},\qquad \textrm{for all $i\neq j$}
\end{array}\right.
\end{displaymath}
has solutions in $\f$.
\end{lem}
\begin{proof}
We prove the statement by induction on $t$.

For $t=2,3$, one can easily check it.

 Now we assume that the result is true for $t\geqslant 3$.

  Then for $t+1$, by induction hypothesis, suppose the equation
\[
   \sum_{1\leqslant i < j\leqslant t}X_{i}X_{j}=b
\]
has a solution $(x_{1},\cdots,x_{t})\in \f^{t}$ such that $x_{i}\neq x_{j}$ for any $i\neq j$.

If $x_{i}\neq 0$ for all $1\leqslant i\leqslant t$, then $(0,x_{1},\cdots,x_{t})\in \f^{t+1}$ is a solution of the equation
\begin{displaymath}
\left\{\begin{array}{rcl}
b & = & \sum_{1\leqslant i < j\leqslant t+1}X_{i}X_{j},\\
X_{i}&\neq & X_{j},\qquad \textrm{for all $i\neq j$.}
\end{array}\right.
\end{displaymath}
Then we finish the induction procedure.

Otherwise, without loss of generality, we assume $x_{1}=0$.
 We want to find some $x'_{t}\in \f\setminus \{x_{1},\cdots,x_{t}\}$ and $x\in  \f\setminus \{x_{1},\cdots,x_{t-1},x'_{t}\}$ such that $(x_{1},\cdots,x_{t-1},x'_{t},x)$ forms a solution of the following equation
\[
   \sum_{ 1\leqslant i < j\leqslant t+1}X_{i}X_{j}=b.
\]
Denote $a=x_{1}+\cdots+x_{t}$. Because permutations of $(x_{1},\cdots,x_{t})$ are also solutions of the equation and $t\geqslant 3$, we may assume $x_{t}\neq a$. Thus
\begin{displaymath}
\left\{\begin{array}{rcl}
a^{2}-2b & = & \sum_{i=1}^{t}x_{i}^2,\\
a&= & x_{1}+\cdots+x_{t}.
\end{array}\right.
\end{displaymath}
It induces
\begin{displaymath}
\left\{\begin{array}{rcl}
a^{2}-2x'_{t}(x_{t}-x'_{t})-(x_{t}-x'_{t})^{2}+x^{2}-2b & = & \sum_{i=1}^{t-1}x_{i}^2+(x'_{t})^{2}+x^{2},\\
a-(x_{t}-x'_{t})+x&= & x_{1}+\cdots+x'_{t}+x.
\end{array}\right.
\end{displaymath}
To find solutions of the equation
\begin{displaymath}
\left\{\begin{array}{rcl}
b & = & \sum_{1\leqslant i < j\leqslant t+1}X_{i}X_{j},\\
X_{i}&\neq & X_{j},\qquad \textrm{for all $i\neq j$}
\end{array}\right.
\end{displaymath}
 in $\f$, it suffices to show the existences of $x'_{t}\in \f\setminus \{x_{1},\cdots,x_{t}\}$ and $x\in  \f\setminus \{x_{1},\cdots,x_{t-1},x'_{t}\}$
which satisfy
\[
  (a-(x_{t}-x'_{t})+x)^{2}=a^{2}-2x'_{t}(x_{t}-x'_{t})-(x_{t}-x'_{t})^{2}+x^{2}.
\]
That is, the solution
\[
   x=\frac{(x_{t}-x'_{t})(x_{t}-a)}{x_{t}-x'_{t}-a}\in \f^{*}
\]
satisfies
\[
  x\in  \f^{*}\setminus \{x_{2},\cdots,x_{t-1},x'_{t}\}.
\]
If
\[
  x= \frac{(x_{t}-x'_{t})(x_{t}-a)}{x_{t}-x'_{t}-a}=x_{i}, \qquad\textrm{for some $i=2,\cdots,t-1$,}
\]
then
\[
   x'_{t}=\frac{(x_{t}-x_{i})(x_{t}-a)}{x_{t}-x_{i}-a}.
\]

If
\[
  x= \frac{(x_{t}-x'_{t})(x_{t}-a)}{x_{t}-x'_{t}-a}=x'_{t},
\]
then it has at most two solutions for $x'_{t}$, say $c_{1},c_{2}$.
So, we can pick any element from
\[
  \f^{*}\setminus \{x_{2},\cdots,x_{t},\frac{(x_{t}-x_{2})(x_{t}-a)}{x_{t}-x_{2}-a},\cdots,\frac{(x_{t}-x_{t-1})(x_{t}-a)}{x_{t}-x_{t-1}-a},c_{1},c_{2}\}
\]
for $x'_{t}$. And under the condition
\[
    q-1>t-1+t-2+2, \qquad\textrm{i.e., $t\leqslant \frac{q-1}{2}$,}
\]
 the equation
\begin{displaymath}
\left\{\begin{array}{rcl}
b & = & \sum_{1\leqslant i < j\leqslant t+1}X_{i}X_{j},\\
X_{i}&\neq & X_{j},\qquad \textrm{for all $i\neq j$},
\end{array}\right.
\end{displaymath}
always has solutions in $\f$.
\end{proof}

 we calculate a sum of quadratic character of a finite field.
\begin{lem}\label{charsum}
Let $q$ be an odd prime power and $\f$ a finite field of $q$ elements. For any $c\in \f$, we have
\begin{displaymath}
\left|\sum_{x\in \f} \eta(x^2+c)\right| = \left\{ \begin{array}{ll}
q-1, & \textrm{if $c=0$,}\\
3, & \textrm{if $\eta(-1)=\eta(c)=-1$,}\\
1, & \textrm{otherwise.}
\end{array} \right.
\end{displaymath}
where $\eta$ is the quadratic character of $\f$, i.e.,
\begin{displaymath}
\eta(x) = \left\{ \begin{array}{ll}
1, & \textrm{if $x$ is a square,}\\
-1, & \textrm{if $x$ is not a square,}\\
0, & \textrm{if $x=0$.}
\end{array} \right.
\end{displaymath}
\end{lem}
\begin{proof}The statement holds obviously for $c=0$.

Now we assume $c\neq 0$, consider the quadratic equation
\[
  x^{2}+c=y^{2},
\]
namely,
\[
 (y-x)(y+x)=c.
\]
It is easy to see that the equation has $(q-1)$ solutions
\[
   \left\{\left(\frac{b-c/b}{2},\frac{b+c/b}{2}\right)\,|\,b\in \f^{*}\right\}.
\]
If $\eta(-1)=\eta(c)=-1$, then $-c$ is a square. So there are two solutions for the equation $x^{2}+c=0$, this means that
\[
\left|\sum_{x\in \f} \eta(x^2+c)\right| = \left|\frac{q-1-2}{2}-(q-\frac{q-1-2}{2})+0+0\right|=3.
\]
For other cases, with the same argument, we finish the rests of this lemma.
\end{proof}
By this lemma, we can compute the number of squares in a finite field whose images under a fixed affine map are still squares. It is interesting that
squares are about one half elements in a finite field,  and about one half squares are still squares under the action of
 an affine map.
\begin{cor}\label{sqrtnum}
Let $q$ be an odd prime power, and $\f$ a finite field of $q$ elements. For any $a, c\in \f^{*}$, we have
\[
   A= \#\{x\in \f\,|\,\textrm{both $x$ and $ax+c$ are squares in $\f$}\}\geqslant \frac{q-1}{4}.
\]
\end{cor}
\begin{proof}
By Lemma~\ref{charsum}, we have
\begin{displaymath}
\left|\sum_{x\in \f} \eta(ax^2+c)\right| =\left|\sum_{x\in \f} \eta(x^2+ac)\right| = \left\{ \begin{array}{ll}
3, & \textrm{if $\eta(-1)=\eta(ac)=-1$,}\\
1, & \textrm{otherwise.}
\end{array} \right.
\end{displaymath}
Let $m=\#\{x\in \f\,|\, \textrm{$ ax^{2}+c$ is a square in $\f$}\}$, then
\begin{displaymath}
\left|\sum_{x\in \f} \eta(ax^2+c)\right|= \left\{ \begin{array}{ll}
|m-2-(q-m)|, & \textrm{if $\eta(-ac)=-1$,}\\
|m-(q-m)|, & \textrm{if $\eta(-ac)\neq -1$.}
\end{array} \right.
\end{displaymath}
Comparing the above two equalities, it follows that
\[
    m\geqslant \frac{q-1}{2}.
\]
Since $A=\frac{m+(\eta(c)+1)/2}{2}\geqslant \frac{m}{2}$, we complete the proof.

\end{proof}

%\begin{prop}[\cite{liwan2}]
%Suppose constants $\epsilon>0$ and $c>0$ satisfy
% \[
 %    \epsilon c^{1/2}+c<\frac{1}{q^{2/t}}-\frac{1}{p},
% \]
% where $p$ is the characteristic of the finite field $\f$.
% If $4\epsilon^{2}\ln^{2}(q)<t\leqslant cq$ and $2<\epsilon t^{1/2}$, then for any $a,b\in \f$, there exist pairwise distinct %$\alpha_{1},\cdots,\alpha_{t}\in\f$ such that
%\[
%  1+ax+bx^{2}\equiv \prod_{j=1}^{t}(1-\alpha_{j}x) \quad \mod x^{3}.
%\]
%\end{prop}
If we only consider $\frac{q}{7}<t\leqslant \frac{q-1}{2}$, taking $c=1/2$ and $\epsilon=\sqrt{2}\left(\frac{1}{q^{14/q}}-\frac{1}{p}-\frac{1}{2}\right)$ in the Theorem~5.3 in \cite{liwan2} where $p$ is the characteristic of the finite field $\f$, then we have

\begin{prop}\label{liwan2}
 If $q\geqslant257$ and $\frac{q}{7}<t\leqslant \frac{q-1}{2}$, then for any $a,b\in \f$, the system of equations
 \begin{displaymath}
\left\{\begin{array}{rcl}
a &=& \sum_{i=1}^{t}X_{i},\\
b & = & \sum_{1\leqslant i < j\leqslant t}X_{i}X_{j},\\
X_{i}&\neq & X_{j},\qquad \textrm{for all $i\neq j$}
\end{array}\right.
\end{displaymath}
 has solutions in $\f$.
 \end{prop}

\begin{thm}\label{dual}
 If $q\geqslant257$ and $\frac{q}{7}<t< \frac{6q}{7}$, then for any $a,b\in \f$, the system of equations
 \begin{displaymath}
\left\{\begin{array}{rcl}
a &=& \sum_{i=1}^{t}X_{i},\\
b & = & \sum_{1\leqslant i < j\leqslant t}X_{i}X_{j},\\
X_{i}&\neq & X_{j},\qquad \textrm{for all $i\neq j$}
\end{array}\right.
\end{displaymath}
 has solutions in $\f$.
 \end{thm}
\begin{proof}The statement is equivalent to that there exist pairwise distinct $\alpha_{1},\cdots,\alpha_{t}\in\f$ such that
\[
  1+ax+bx^{2}\equiv \prod_{j=1}^{t}(1-\alpha_{j}x) \quad \mod x^{3}.
\]

By Proposition~\ref{liwan2}, it suffices to prove that if $\frac{q+1}{2}\leqslant t< \frac{6q}{7}$, for any $a,b\in\f$, there exist pairwise distinct $\alpha_{1},\cdots,\alpha_{t}\in\f$ such that
\[
  1+ax+bx^{2}\equiv \prod_{j=1}^{t}(1-\alpha_{j}x) \quad \mod x^{3}.
\]
So it is equivalent to
\[
  \frac{1-x^{q-1}}{1+ax+bx^{2}}\equiv \frac{1-x^{q-1}}{\prod_{j=1}^{t}(1-\alpha_{j}x)} \quad \mod x^{3}.
\]
Set
\[
   \{\beta_{1},\cdots,\beta_{q-t}\}=\f\setminus\{\alpha_{1},\cdots,\alpha_{t}\}.
\]
Then it is equivalent to the existence of distinct $\beta_{1},\cdots,\beta_{q-t}\in\f$ satisfying
\[
  \frac{1-x^{q-1}}{1+ax+bx^{2}}\equiv \prod_{j=1}^{q-t}(1-\beta_{j}x) \qquad \mod x^{3}.
\]
By Proposition~\ref{liwan2}, when $\frac{q}{7}<q-t\leqslant \frac{q-1}{2}$, i.e.,
\[
   \frac{q+1}{2}\leqslant t< \frac{6q}{7},
\]
the required $\beta_{1},\cdots,\beta_{q-t}$ exist.

\end{proof}
Now we prove Theorem~\ref{main}.

\begin{flushleft}
\textbf{Proof of Theorem~\ref{main}.}
\end{flushleft}
~~~(a) By Lemmas~\ref{red}-\ref{thm}, when $\frac{q+1}{2}\leqslant t\leqslant q-2$, for any $a,b\in\f$, the equation
\begin{displaymath}
\left\{\begin{array}{rcl}
b & = & \sum_{1\leqslant i < j\leqslant t}X_{i}X_{j}+\sum_{i=1}^{t}X_{i}^{2}-a\sum_{i=1}^{t}X_{i},\\
X_{i}&\neq & X_{j},\qquad \textrm{for all $i\neq j$}
\end{array}\right.
\end{displaymath}
has solutions in $\f$.

(b) For $3\leqslant t< \frac{q-11}{6}$, we prove the statement by induction on $t$.

If $t=3$, first we fix an element $x_{1}\in \f$ such that
\[
   2x_{1}^{2}-3b\neq 0.
\]
Now we need to find $x_{2}\in\f\setminus \{x_{1}\}$ such that the equation on $X$
\[
   X^{2}+(x_{1}+x_{2})X+x_{1}^{2}+x_{2}^{2}+x_{1}x_{2}-b=0
\]
has a solution $x_{3}\in\f\setminus \{x_{1},x_{2},-x_{1}-x_{2}\}$.

It is well-known that the above equation on $X$ has solutions if and only if the discriminant
 \begin{displaymath}
\begin{array}{rcl}
\Delta & =  & (x_{1}+x_{2})^{2}-4(x_{1}^{2}+x_{2}^{2}+x_{1}x_{2}-b) \\
 & = & -3x_{2}^{2}-2x_{1}x_{2}-3x_{1}^{2}+4b
\end{array}
\end{displaymath}
is a square in $\f$.

Note that the characteristic $p\neq 3$, so the discriminant
$$
\Delta=-3(x_{2}+\frac{x_{1}}{3})^{2}-\frac{8x_{1}^{2}}{3}+4b.
$$
Denote
\[
   x_{2}(\Delta)=\{x_{2}\in\f\setminus \{x_{1}\}\,|\,\textrm{the induced $\Delta$ is a square}\}.
\]
Note that $x_{2}\neq x_{1}$, by Corollary~\ref{sqrtnum}, we have
\[
  \left| x_{2}(\Delta)\right|\geqslant \left(2\left(\frac{q-1}{4}-1\right)+1\right)-1=\frac{q-5}{2}.
\]
Since the equation on $x_{2}$
\[
    X^{2}+(x_{1}+x_{2})X+x_{1}^{2}+x_{2}^{2}+x_{1}x_{2}-b=0
\]
gives at most two solutions $x_{i}^{(1)}, x_{i}^{(2)}$ for each $x=x_{i}$ ($i=1,2$) or $c^{(1)}, c^{(2)}$ for $x=-x_{1}-x_{2}$, respectively.

From the argument above, we can pick any element from the set
\[
   x_{2}(\Delta)\setminus \{x_{1}^{(1)}, x_{1}^{(2)},x_{2}^{(1)}, x_{2}^{(2)},c^{(1)}, c^{(2)}\}
\]
for $x_{2}$. And this can be done under the assumption
\[
     6<\frac{q-5}{2},\qquad \textrm{ i.e., $q>17$}.
\]
%For small $q$, $5< q\leqslant 17$, by computing with the program Magma one can check that the statement also holds.

 Now we assume that the result is true for $t\geqslant 3$.

Then for $t+1$, by induction hypothesis, we may assume $(x_{1},\cdots,x_{t})\in \f^{t}$ forms a solution of the equation
\begin{displaymath}
\left\{\begin{array}{rcl}
b & = & \sum_{1\leqslant i < j\leqslant t}X_{i}X_{j}+\sum_{i=1}^{t}X_{i}^2,\\
X_{i}&\neq & X_{j},\qquad \textrm{for all $i\neq j$,}
\end{array}\right.
\end{displaymath}
and
\[
a=\sum_{i=1}^{t}x_{i}\neq 0.
\]
 Then
\begin{displaymath}
\left\{\begin{array}{rcl}
2b-a^{2} & = & \sum_{i=1}^{t}x_{i}^2,\\
a&=&\sum_{i=1}^{t}x_{i}.
\end{array}\right.
\end{displaymath}
The same as the proof in Lemma~\ref{thm}, if $x_{i}\neq 0$ for all $1\leqslant i\leqslant t$, then $(x_{1},\cdots,x_{t},0)\in \f^{t}$ forms a solution of the equation
\begin{displaymath}
\left\{\begin{array}{rcl}
b & = & \sum_{1\leqslant i < j\leqslant t+1}X_{i}X_{j}+\sum_{i=1}^{t+1}X_{i}^2,\\
X_{i}&\neq & X_{j},\qquad \textrm{for all $i\neq j$,}
\end{array}\right.
\end{displaymath}
Otherwise, assume $x_{t}=0$. We replace $x_{t}$ by a proper element $x'_{t}\in \f^{*}$ to induce an element $x\in \f$ such that
$(x_{1},\cdots,x'_{t},x)\in \f^{t+1}$ is a solution of the equation
\begin{displaymath}
\left\{\begin{array}{rcl}
b & = & \sum_{1\leqslant i < j\leqslant t+1}X_{i}X_{j}+\sum_{i=1}^{t+1}X_{i}^2,\\
X_{i}&\neq & X_{j},\qquad \textrm{for all $i\neq j$,}
\end{array}\right.
\end{displaymath}
and
\[
 \sum_{i=1}^{t-1}x_{i}+x'_{t}+x\neq 0.
\]
Firstly, we pick $x'_{t}\in \f\setminus\{x_{1},\cdots,x_{t}\}$, which will be determined. Then by the induction hypothesis,
\begin{displaymath}
\left\{\begin{array}{rcl}
2b-a^{2}+(x'_{t})^{2}+x^{2} & = & \sum_{i=1}^{t-1}x_{i}^2+(x'_{t})^{2}+x^{2},\\
a+x'_{t}+x&= & x_{1}+\cdots+x'_{t}+x.
\end{array}\right.
\end{displaymath}
We want to find some $x'_{t}\in \f\setminus\{x_{1},\cdots,x_{t},-a\}$ such that the equation on $x$
\[
    (a+x'_{t}+x)^{2}=a^{2}-(x'_{t})^{2}-x^{2}
\]
 gives a solution $x\in \f\setminus\{x_{1},\cdots,x'_{t},-(\sum_{i=1}^{t-1}x_{i}+x'_{t})\}$.
The equation is
\[
   x^{2}+(a+x'_{t})x+x'_{t}(a+x'_{t})=0.
\]
Here, we see that why we require $x'_{t}\neq -a$ at the beginning. Indeed, if $x'_{t}= -a$, then the equation above is reduced to
\[
    x^{2}=0,
\]
which has only zero solution $x=0$. Then
\[
   \sum_{i=1}^{t-1}x_{i}+x'_{t}+x= 0,
\]
which goes against our requirement in the induction.

Similarly as the proof for the case $t=3$, the above equation on $x$ has solutions if and only if the discriminant
 \begin{displaymath}
\begin{array}{rcl}
\Delta & =  & (a+x'_{t})^{2}-4x'_{t}(a+x'_{t}) \\
 & = & -3(x'_{t})^{2}-2ax'_{t}+a^{2}
\end{array}
\end{displaymath}
is a square in $\f$.

Note that the characteristic $p\neq 3$, so the discriminant
$$
\Delta=-3(x'_{t}+\frac{a}{3})^{2}+\frac{4a^{2}}{3}.
$$
Denote
\[
   x'_{t}(\Delta)=\{x'_{t}\in\f^{*}\,|\,\textrm{the induced $\Delta$ is a square}\}.
\]
Note that $x'_{t}\neq 0$, by Corollary~\ref{sqrtnum}, we have
\[
  \left| x'_{t}(\Delta)\right|\geqslant \left(2\left(\frac{q-1}{4}-1\right)+1\right)-1=\frac{q-5}{2}.
\]
Since the equation on $x'_{t}$
\[
   x^{2}+(a+x'_{t})x+x'_{t}(a+x'_{t})=0
\]
gives at most two solutions $x_{i}^{(1)}, x_{i}^{(2)}$ for each $x=x_{i}$ ($i=1,2,\cdots,t-1$) or $c^{(1)}, c^{(2)}$ for $x=-(\sum_{i=1}^{t-1}x_{i}+x'_{t})$ or $(x'_{t})^{(1)}, (x'_{t})^{(2)}$ for $x=x'_{t}$, respectively.

From the argument above, we pick any element $x'_{t}$ from the set
\[
    x'_{t}(\Delta)\setminus \{x_{1},\cdots,x_{t},x_{1}^{(1)}, x_{1}^{(2)},\cdots,x_{t-1}^{(1)}, x_{t-1}^{(2)},(x'_{t})^{(1)}, (x'_{t})^{(2)},c^{(1)}, c^{(2)},-a\}
\]
and finish the induction. And under the assumption
\[
     3t+3<\frac{q-5}{2},\qquad \textrm{ i.e., $t<\frac{q-11}{6}$},
\]
such an $x'_{t}$ always exists.

Since
\[
       \frac{q}{7}< \frac{q-11}{6}
\]
as $q\geqslant257$, by combining (a), (b) and Theorem~\ref{dual}, we complete the proof of Theorem~\ref{main}.
 \hfill$\Box$

\end{document}